\RequirePackage[l2tabu, orthodox]{nag}		
\documentclass[10pt]{article}

\usepackage[T1]{fontenc}
\usepackage{lmodern}
\usepackage{microtype}

\usepackage{natbib}       									
\usepackage{amsmath}                        
\numberwithin{equation}{section}
\usepackage{graphicx}                       
\usepackage{epstopdf}  
\usepackage{subfigure}
\usepackage{enumitem}                       
\usepackage{mdwlist}												
\usepackage{ulem}														
\usepackage{cancel}													
\usepackage[all]{xy}
\usepackage[plainpages=false, pdfpagelabels]{hyperref} 
\hypersetup{
		colorlinks   = true,
		citecolor    = black,
		linkcolor    = Maroon,
		urlcolor     = Turquoise
}
\usepackage[dvipsnames]{xcolor}
\usepackage{amssymb}                        
\usepackage[mathscr]{eucal}                 
\usepackage{dsfont}													
\usepackage[paperwidth=8.5in,paperheight=11in,top=1.25in, bottom=1.25in, left=1.00in, right=1.00in]{geometry}
\usepackage{mathtools}                      
\mathtoolsset{showonlyrefs=true}            
\usepackage{fixltx2e,amsmath}               
\MakeRobust{\eqref}

\usepackage{microtype}											
\usepackage{amsthm}                         
\allowdisplaybreaks                         
\theoremstyle{plain}
\newtheorem{theorem}{Theorem}
\numberwithin{theorem}{section}
\newtheorem{lemma}[theorem]{Lemma}          

\theoremstyle{definition}

\newtheorem{remark}[theorem]{Remark}
\newtheorem{assumption}[theorem]{Assumption}
\newtheorem{step}{Step}

\renewcommand{\(}{\left(}
\renewcommand{\)}{\right)}
\renewcommand{\[}{\left[}
\renewcommand{\]}{\right]}

\newcommand\Cb{\mathds{C}}
\newcommand\Eb{\mathds{E}}

\newcommand\Pb{\mathds{P}}

\newcommand\Rb{\mathds{R}}

\newcommand\Ac{\mathscr{A}}

\newcommand\Fc{\mathscr{F}}
\newcommand\Gc{\mathscr{G}}

\newcommand\Oc{\mathscr{O}}

\newcommand\Qc{\mathscr{Q}}

\newcommand\om{\omega}

\newcommand\sig{\sigma}

\newcommand\gam{\gamma}
\newcommand\Gam{\Gamma}
\newcommand\lam{\lambda}
\newcommand\del{\delta}

\newcommand\vb{\bar{v}}
\newcommand\ub{\bar{u}}

\newcommand\varphih{\widehat{\varphi}}

\renewcommand\d{\partial}

\newcommand\ii{\mathtt{i}}
\newcommand\dd{\mathrm{d}}
\newcommand\ee{\mathrm{e}}

\newcommand\Id{\textrm{Id}}

\newcommand{\trunc}{h}
\newcommand{\truncz}{h(z)}

\newcommand{\levy}{L\'{e}vy }

\newcommand{\half}{\frac{1}{2}}
\newcommand{\compensator}{\nu}

\newcommand{\const}{C}
\newcommand{\sdedrift}{b}
\DeclareMathOperator{\ess}{ess}
\newcommand{\moment}{\mu}
\newcommand{\expmoment}{\varphi}
\newcommand{\jumpintegrand}{c}

\begin{document}

\title{Pricing Variance Swaps on Time-Changed Markov Processes}

\author{
Peter Carr
\thanks{Department of Finance and Risk Engineering, NYU Tandon, New York, USA.}
\and
Roger Lee
\thanks{Department of Mathematics, University of Chicago, Chicago, USA.}
\and
Matthew Lorig
\thanks{Department of Applied Mathematics, University of Washington, Seattle, USA.}
}

\date{This version: \today}

\maketitle

\begin{abstract}
We prove that the variance swap rate (fair strike) equals the price of a co-terminal European-style contract when the underlying is an exponential Markov process, time-changed by an arbitrary continuous stochastic clock, which has arbitrary correlation with the driving Markov process, provided that the payoff function $G$ of the European contract satisfies an ordinary integro-differential equation, which depends only on the dynamics of the Markov process, not on the clock.  We present examples of Markov processes where the function $G$ that prices the variance swap can be computed explicitly.  In general, the solutions $G$ are not contained in the logarithmic family previously obtained in the special case where the Markov process is a \levy process.
\end{abstract}


%
%

\section{Introduction}
\label{sec:intro}
Consider a forward price $F$ that evolves in continuous time.  Let time zero be the valuation time for a derivative security written on the path of $F$, with a fixed maturity date $T > 0$.  Assume that $F_0 > 0$ is a known constant, and that the $F$ process is strictly positive over a time interval $[0,T]$.  As a result, the $\log$ price process $X:=\log F$ is well-defined, and derivative securities expiring at $T$ can also be written on the path of $X$. In particular, we focus on a continuously-monitored variance swap, which pays the difference between the terminal quadratic variation of the $\log$ price process $[\log F]_T$ and a constant determined at inception.  For brevity, we will refer to a continuously monitored variance swap as a VS in the sequel.  As with any swap, the constant that is determined at inception is chosen so that there is no initial cost of entering into the VS.  The objective of this paper is to give additional conditions on the dynamics of $F$ under which this constant can be determined from an initial observation of the $T$-maturity implied volatility smile.

Earlier papers by \cite{neuberger} and \cite{dupire1993model} show that continuity of $F$ suffices for pricing a VS relative to the co-terminal smile.  \cite{carr2011variance} weakens the continuity hypothesis by showing that the $\log$ price $X$ can be specified as a L\'{e}vy process running on an unspecified continuous clock.  When the L\'evy process is specified as Brownian motion with drift $(-1/2)$, the earlier results of \cite{neuberger} and \cite{dupire1993model} arise as a special case.  The more general formulation of \cite{carr2011variance} allows for the variance and jump-intensity to depend on the level of $X$ through a local time-change (see Remark \ref{rmk:tc}).  However, the local variance and L\'evy kernel must have the same functional dependence on $X$ (up to a scaling constant).  Additionally, while the arrival rate of each jump size in $X$ is allowed to depend on the level of $X$, the ratio of the arrival rates at any two jump sizes is constant in that previous paper.

This paper weakens the stationary independent increments property of the L\'{e}vy process used by \cite{carr2011variance}.  We allow that $X$ could be specified as a time-homogeneous Markov process running on an unspecified continuous clock.  As a result (i) the variance and jump-intensity may have distinct $X$-dependence and (ii) the ratio of the arrival rates at any two jump sizes of $X$ can depend on the current level of $X$.

In effect, we allow the background process to have nearly the full generality of \emph{general} Markov processes whose jump times are not predictable, as discussed in Remark \ref{rem:cinlar}.  We allow that general background Markov process to undergo a time-change by an unspecified continuous stochastic clock which may have arbitrary correlation or dependence on the background process.   In this setting, we prove that European-style payoff functions $G$ price the variance swap, in the sense that the variance swap rate (fair strike) equals the price of a contract paying $G(\log F_T) - G(\log F_0)$, 
provided that $G$ satisfies an ordinary integro-differential equation that depends only on the dynamics of the Markov driver, not on the clock.

Our results are related to the semiparametric approach taken by \cite{lorigmendoza}, who consider the pricing of a VS when the underlying forward price $F$ is modeled as Feller diffusion time-changed by an unspecified L\'evy subordinator.
For fully parametric approaches to VS pricing in models with jumps and stochastic volatility we refer the reader to 
\cite{itkin2010pricing,zheng-kwok,FILIPOVIC201644,CUI2017381}.
For model-independent bounds on (discrete and continuous) VS prices see \cite{hobsonklimmek}, \cite{Nabil}, and 
\citet[Example 5.7]{Henry-Labordere2016}.


The rest of this paper proceeds as follows.
Section \ref{sec:tc.dynamics} specifies dynamics for the forward price process and verifies that these dynamics can arise from time-changing the solution of a stochastic differential equation.
Section \ref{sec:pricing} states and proves our main result (Theorem \ref{thm:main}), which establishes
that the VS has the same value as a European-style claim whose payoff function solves an ordinary integro-differential equation (OIDE).
Section \ref{sec:examples} provides examples of price dynamics for which we can solve the OIDE explicitly.
Section \ref{sec:conclusion} concludes.

%
%

\section{Time-changed Markov dynamics}
\label{sec:tc.dynamics}


\subsection{Assumptions}
\label{sec:assumptions}
With respect to a (``calendar-time'') filtration $\{\Fc_t\}_{t\geq 0}$ on a probability space $(\Omega, \Fc, \Pb)$, assume that $X$ is a semimartingale with predictable characteristics $(B,A,\nu)$, relative to a truncation function $\trunc$ (to be definite, let $\truncz:=z\mathbf{1}_{\{|z|\leq 1\}}$), which satisfy
\begin{align}
B_t
	&= \int_0^t b_\trunc(X_{s-}) \dd \tau_s , &
A_t
	&= \int_0^t a^2(X_{s-}) \dd \tau_s , &
\compensator(\dd t,\dd z)
	&=	\dd \tau_t\times\mu(X_{t-},\dd z) ,  \label{eq:semichar}
\end{align}
where $\tau$ is a real-valued continuous increasing adapted process that is null at zero, $a$ is a Borel function, and for each fixed $x \in \Rb$ the $\mu(x,\cdot)$ is a \levy measure, and
\begin{align}
\sup_{x \in \Rb}|a(x)|
	&<	\infty , & 
\sup_{x \in \Rb} \int_\Rb z^2 \mu(x,\dd z)
	&<	\infty , &  
\sup_{x \in \Rb} \int_\Rb (\ee^z-1-z) \mu(x,\dd z)
	&<	\infty , \label{eq:asqexpintegrable}  
\end{align}
with
\begin{align}
b_\trunc(x)
	&:=		-\frac{1}{2}a^2(x) - \int_\Rb \( \ee^z - 1 - \truncz \) \mu(x,\dd z) . \label{eq:drift}
\end{align}
The intuition of the \textit{L\'evy kernel} or \textit{transition kernel} $\mu$ is that it assigns, to each point $x$ in the state space, a ``local'' \levy measure $\mu(x,\cdot)$.  Jumps of size in any interval $J$ arrive with intensity $\mu(x,J)$ when $X$ is at $x$.

Define the underlying forward price process $F=\{F_t\}_{t \in [0,T]}$ by
\begin{align}
F_t
	&= 		\exp(X_t).  \label{eq:F}
\end{align}
Regarding $\Pb$ as risk-neutral measure, we have chosen $b_\trunc$ in \eqref{eq:drift} to ensure $F$ is a local martingale.  If $\tau_T$ is integrable, then Lemma \ref{lem:martingale} will imply that $F$ is a true martingale.


\subsection{Time-change of an SDE solution}
\label{sec:tcsde}
This section verifies that the assumptions of Section \ref{sec:assumptions} hold in the case that $X$ comes from time-changing the solution of a stochastic differential equation (SDE) driven by a Brownian motion and a Poisson random measure.
With respect to a filtration $\{\Gc_u\}_{u\geq 0}$ (the ``business time'' filtration), consider a Brownian motion $W$, and a Poisson random measure $N$ with intensity measure $\mu_N(\dd z)\dd u$ for some \levy measure $\mu_N$.  Assume that $Y$ is a semimartingale that satisfies
\begin{align}
\dd Y_u
	&=	\sdedrift(Y_u) \, \dd t + a(Y_u) \, \dd W_u + \int_{z\in\Rb} \jumpintegrand(Y_{u-},z)\, (N(\dd u,\dd z)-
\mu_N(\dd z)\dd u),
\label{eq:dY}
\end{align}
where $a$ is a bounded Borel function,
$b$ is given by
\begin{align}
\sdedrift(x) &= -\frac{1}{2}a^2(x) - \int_\Rb \(\ee^z - 1 - z \) \mu(x,\dd z),
\end{align}
and $\jumpintegrand$ is a Borel function such that $\mu$, defined for each Borel set $J$ by
\begin{align}
\mu(x,J) &:= \mu_N(\{z: \jumpintegrand(x,z)\in J\backslash\{0\}\}),
\end{align}
satisfies
\begin{align}
\sup_{x \in \Rb} \int_\Rb z^2\mu(x,\dd z) + \sup_{x \in \Rb} \int_\Rb (\ee^z-1-z)\mu(x,\dd z)
	&<	\infty.
\end{align}
Then by \citet[Prop.\ III.2.29]{jacod1987limit}, the semimartingale characteristics of $Y$ are $(\widetilde{B},\widetilde{A},\widetilde{\compensator})$, where
\begin{align}
\widetilde{B}_u
	&= \int_0^u b_\trunc(Y_{v-})\dd v , & 
\widetilde{A}_u
	&= \int_0^u a^2(Y_{v-})\dd v , & 
\widetilde{\compensator}(\dd u,\dd z)
	&=\dd u\times\mu(Y_{u-},\dd z), \label{eq:abnutilde} 
\end{align}
with $b_\trunc$ defined in \eqref{eq:drift}.

Now let $\{\tau_t\}_{t\geq 0}$ be a continuous increasing family of finite $\Gc$-stopping times (which are \textit{not}
assumed to be independent of $Y$).  Let the ``calendar-time'' filtration be defined by $\Fc_t:=\Gc_{\tau_t}$, and let
\begin{align}
X_t
	&:= Y_{\tau_t}.
\end{align}
By \citet[Lemma 5]{kallsen2002time}, the $\Fc$-characteristics of $X$ are $(B,A,\nu)$ where $A_t = \widetilde{A}_{\tau_t}$, $B_t = \widetilde{B}_{\tau_t}$ and $\nu$ is determined by
\begin{align}
\int_{[0,t]\times\Rb} \mathbf{1}_J(z) \nu(\dd s,\dd z)
	&= \int_{[0,\tau_t]\times\Rb} \mathbf{1}_J(z) \widetilde{\nu}(\dd u,\dd z) , \label{eq:tccompensator}
\end{align}
for general Borel sets $J$ and $t \geq 0$.  By the first two equalities in \eqref{eq:abnutilde} we have
\begin{align}
\widetilde{A}_{\tau_t}
	&= \int_0^{\tau_t} a^2(Y_{v-})\dd v = \int_0^t a^2(X_{s-})\dd\tau_s, &
\widetilde{B}_{\tau_t}
	&= \int_0^{\tau_t} b_\trunc(Y_{v-})\dd v = \int_0^t b_\trunc(X_{s-})\dd\tau_s,
\end{align}
and, by substituting
the last equality in \eqref{eq:abnutilde}
into \eqref{eq:tccompensator} and changing variables $u$ to $\tau_s$, we have
\begin{align}
\int_{[0,t]\times\Rb} \mathbf{1}_J(z) \nu(\dd s,\dd z)
	&= \int_{[0,t]}\int_{\Rb} \mathbf{1}_J(z) \mu(X_{s-},\dd z)\dd \tau_s.
\end{align}
Therefore $(B,A,\nu)$ satisfy \eqref{eq:semichar}.  This verifies the hypotheses of Section \ref{sec:assumptions}, as claimed.

\begin{remark}\label{rem:cinlar}
Time-changes of SDE solutions are nearly as general as time-changes of \emph{general} Markov processes whose jump times are not predictable.

To be precise, \cite{cinlarjacod} show that every strong Markov quasi-left-continuous 
semimartingale (which includes every Feller semimartingale) is a continuous time change of an SDE solution driven by Brownian motion and a Poisson random measure
(on an enlarged probability space if needed).
Thus, if $X$ is a continuous time-change $\tau'$ of a \emph{general} Feller semimartingale $Y'$, then by \c{C}inlar-Jacod, $Y'$ is a continuous time change $\tau''$ of an SDE solution $Y$, and therefore $X$ is a continuous time change $\tau'\circ\tau''$ of an SDE solution $Y$.
\end{remark}


\subsection{Notations}
\label{sec:notation}
Let $C^n(\Rb)$ denote the class of $n$-times continuously differentiable functions, and define the integro-differential operator $\Ac$ by
\begin{align}
\Ac g(x)
	&:= b_\trunc(x) g'(x)+\frac{a^2(x)}{2} g''(x)+\int_\Rb(g(x+z)-g(x)-g'(x)\truncz)\ \mu(x,\dd z)\\
	&= \frac{a^2(x)}{2}(g''(x)-g'(x))+\int_\Rb(g(x+z)-g(x)+(1-\ee^z)g'(x))\ \mu(x,\dd z) , \label{eq:generator}
\end{align}
for all $g\in C^2(\Rb)$ such that $g(x+z)-g(x)+(1-\ee^z)g'(x)\in L^1(\mu(x,\dd z))$ for all $x$.

In more concise notation,
\begin{align}
\Ac
	&= \frac{1}{2} a^2(x) \( \d^2 - \d \) + \int_\Rb \Big(\ee^{z\d} - 1 +(1-\ee^z) \d \Big)\ \mu(x,\dd z), \label{eq:A.first}
\end{align}
where $\ee^{z\d}$ is the \textit{shift operator} defined by $\ee^{z\d} g(x) := g(x+z)$.  This use of $\d$ to express translations in the jump part of the generator $\Ac$ follows \cite{itkin2012using}.

Let $C^{1+}(\Rb)$ denote the union of $C^2(\Rb)$ and the following set: all $C^1(\Rb)$ functions $g$ whose derivative is everywhere absolutely continuous, and whose second derivative (which therefore exists a.e.) is equal (a.e.) to a bounded function, which we will still denote by $g''$ or $\d^2g$.

Thus the definition of $\Ac$ extends, by relaxing the $g\in C^{2}(\Rb)$ condition to $g\in C^{1+}(\Rb)$, which still defines $\Ac g$ uniquely, up to sets of measure zero, via \eqref{eq:generator}.

%
%

\section{Variance swap pricing}
\label{sec:pricing}
In what follows, each $C$ will denote a constant (non-random and non-time-varying).  Different instances of $C$, even in the same expression, may have different values.

\begin{lemma}\label{lem:special}
Suppose that $g\in C^{1+}(\Rb)$ and there exists $p\in\Rb$ such that
\begin{align}
\sup_{x\in\Rb}|g'(x)\ee^{-px}|
	&<\infty &
	&\textrm{and}&
\sup_{x\in\Rb}\int_\Rb (\ee^{pz}-1-pz)\ \mu(x,\dd z)
	&<\infty . \label{eq:expcond}
\end{align}
Then $g(X)$ is a special semimartingale.
\end{lemma}
\begin{proof}
By the form of It\^o's rule in, for instance \citet[Theorem IV.70]{protter},
$g(X)$ is a semimartingale.

By \citet[Lemma 2.8]{kallsen2002esscher},
it suffices to check that the predictable process
\begin{align}
\int_0^t \int_{\{z: |g(X_{s-}+z)-g(X_{s-})|>1\}}|g(X_{s-}+z)-g(X_{s-})|\mu(X_{s-},\dd z)\dd\tau_s  \label{eq:Gbigjumps}
\end{align}
is finite (hence of finite variation, as it is increasing in $t$).

In the case $p=0$, we have $|g(x+z)-g(x)|\leq \const|z|$.  In the case $p\neq 0$, we have
\begin{align}
|g(x+z)-g(x)|\leq \int_{x\wedge(x+z)}^{x \vee (x+z)} \const \ee^{p\zeta}\dd \zeta = \const \ee^{px}|\ee^{pz}-1|.
\end{align}
In this case, for each $m>0$, let $k(m)$ be such that $|\ee^{pz}-1|\textbf{1}_{|\ee^{pz}-1|>1/m}<(\ee^{pz}-1-pz)+k(m)z^2$ for all $z$,
and let $M:=\sup_{s\in[0,T]} \ee^{pX_s}<\infty$ because $X$ is c\`adl\`ag.  Then
\begin{align}
\int_{\{z: |g(X_{s-}+z)-g(X_{s-})|>1\}}|g(X_{s-}+z)-g(X_{s-})|\mu(X_{s-},\dd z)
\end{align}
is bounded in case $p=0$ by $\sup_{x\in\Rb}\int_{\{z: |z|>1/\const\}}\const |z|\mu(x,\dd z)<\infty$,
and in case $p\neq 0$ by $C$ times
\begin{align}
&\sup_{x\in\Rb}\int_{\{z: |\ee^{pz}-1|>1/(CM)\}} M |\ee^{pz}-1|\mu(x,\dd z) \\ 
 &\quad \leq M \sup_{x\in\Rb}\int_\Rb (\ee^{pz}-1-pz)\mu(x,\dd z)+Mk(CM) \sup_{x\in\Rb}\int_\Rb z^2 \mu(x,\dd z)<\infty.
\end{align}
These upper bounds do not depend on $s\in[0,t]$, which verifies that \eqref{eq:Gbigjumps} is finite.
\end{proof}

\begin{lemma}\label{lem:supx}
If $\Eb\tau_T<\infty$ then $\Eb\sup_{t\in[0,T]}|X_t|<\infty$.
\end{lemma}
\begin{proof}
Let $B'_t:=B_t +\int_{[0,t]\times\mathbb{R}} (z-h(z)) \nu(\dd u,\dd z)$.
We have $\Eb\sup_{t\in[0,T]}|B'_t|  <\infty$ due to \eqref{eq:asqexpintegrable}  and $\Eb\tau_T<\infty$.

Defining $M_t$ by
\begin{align}
X_t
	&= X_0 + M_t + B'_t ,
\end{align}
we have,
by \citet[Proposition II.2.29]{jacod1987limit},
that $M$ is a local martingale satisfying
\begin{align}
\Eb[M,M]_T
	&= \Eb\int_0^T a^2(X_{s})\dd\tau_s+\Eb\int_0^T\int_\Rb z^2 \mu(X_{s-},\dd z)\dd\tau_s
	< 	\infty ,
\end{align}
because $\Eb\tau_T<\infty$.  By Burkholder-Davis-Gundy, $\Eb\sup_{t\in[0,T]}|M_t|<\infty$, which implies the result.
\end{proof}

\begin{lemma}\label{lem:supexp}
Suppose $\tau_T$ is bounded and $p\in\Rb$ satisfies
\begin{align}\label{eq:pmoment}
\sup_{x\in\Rb}\int_{\Rb} (\ee^{pz}-1-pz) \mu(x,\dd z)
	&< \infty.
\end{align}
Let
\begin{align}
Z_t
	&:= \exp(pX_t-K_t) , \\
K_t
	&:= \int_0^t\half (p^2-p) a^2(X_s)\dd\tau_s+\int_0^t\int_{\Rb}[(\ee^{pz}-1-pz)-p(\ee^{z}-1-z)]\mu(X_{s-},\dd z)\dd\tau_s .
\end{align}
Then $Z$ is a martingale, and
\begin{align}
\Eb \sup_{t\in[0,T]} \exp(pX_t)
	&<	\infty . \label{eq:supl1}
\end{align}
\end{lemma}
\begin{proof}
Let $N$ be the integer-valued random measure associated with the jumps of $X$.  Let $\widetilde{N}:=N-\nu$.

By \citet[Theorem 2.19]{kallsen2002esscher},
the process $Z$ is the stochastic exponential of the local martingale
\begin{align}
pX_t^c+\int_{[0,t]\times\Rb} (\ee^{pz}-1)\widetilde{N}(\dd s,\dd z),
\end{align}
where $X^c$ is the continuous martingale part of $X$.  
By the boundedness of $\tau_T$ and assumptions \eqref{eq:asqexpintegrable} and \eqref{eq:pmoment}, it follows that
\begin{align}
p^2\int_0^T a^2(X_s)\dd\tau_s + \int_0^T\int_{\Rb} (\ee^{pz}-1)^2\wedge(\ee^{pz}-1)\ \mu(X_{s-},\dd z) \dd \tau_s
\end{align}
is bounded.  So by \cite{lepingle1978integrabilite}, the process $Z$ is a martingale and $\Eb\sup_{t\in[0,T]} Z_t<\infty$,
which implies \eqref{eq:supl1} because $\sup_{t\in[0,T]}K_t$ is bounded.
\end{proof}

Let us define two conditions that may be satisfied by $(\tau_T,g)$ where $g\in C^{1+}(\Rb)$.  The first is
\begin{align}\label{eq:gtlin}
\Eb\tau_T<\infty\text{ and }\sup_{x\in\Rb}|g'(x)|+\ess\sup_{x\in\Rb}|g''(x)|<\infty,
\end{align}
and the second is
\begin{equation}\label{eq:gtexp}
\tau_T  \text{ is bounded, and $\exists p\in\Rb$ with }
\sup_{x\in\Rb}\int_{\Rb} (\ee^{pz}-1-pz) \mu(x,\dd z) +\ess\sup_{x\in\Rb}\ee^{-px}(|g(x)|+|g'(x)|+|g''(x)|)<\infty.
\end{equation}
\begin{lemma}\label{lem:martingale}
Assume that $g$ is a sum of finitely many $C^{1+}(\Rb)$ functions, each of which satisfies \eqref{eq:gtlin} or \eqref{eq:gtexp}.  Let
\begin{align}
\Gamma_t
	&:= g(X_t)-g(X_0)-\int_0^t\Ac g(X_{s-})\dd\tau_s , &
	&t \in [0,T].
\end{align}
Then $\Gam$ is a martingale.
\end{lemma}

\begin{proof}
We prove for the case that the $g$ satisfies \eqref{eq:gtlin} or \eqref{eq:gtexp}.  The case that $g$ is the sum of such functions follows immediately by linearity.

Either one of the conditions \eqref{eq:gtlin} or \eqref{eq:gtexp} implies that $\Ac g$ is well-defined.

To show that $\Gamma$ is a local martingale, note that
\citet[Theorem II.2.42c]{jacod1987limit}
extends as follows.  They assume $g$ bounded, only to show that $g(X)$ is a special semimartingale, but the conditions in
Lemma \ref{lem:special} suffice for that conclusion.  Moreover they assume $g\in C^2$,
only to use It\^o's lemma, but $C^{1+}$ suffices here, by
\citet[Theorem IV.70]{protter}
and its first corollary.

To show that $\Gamma$ is a true martingale, it suffices,
by \citet[Theorem I.51]{protter},
to show that $\Eb\sup_{t\in[0,T]}|\Gamma_t|<\infty$.
In case \eqref{eq:gtlin}, let $p:=0$.  In both cases, by
\eqref{eq:asqexpintegrable},
we have
\begin{align}
|g'(x)|\int_\Rb (\ee^z-1-z)\mu(x,\dd z)
	&< \const \ee^{px} , \label{eq:b1}
\end{align}
and by Taylor's theorem and $|g''(x+z)|\leq C\ee^{px+|p|}$ for $|z|<1$, we have
\begin{align}
\int_{|z|<1} |g(x+z)-g(x)-g'(x)z| \mu(x,\dd z)\leq\const \ee^{px+|p|}\int_{|z|<1} z^2 \mu(x,\dd z)
	&\leq \const \ee^{px} , \label{eq:b2}
\end{align}
and by \eqref{eq:asqexpintegrable},
\begin{align}
\int_{|z|>1} |g(x+z)-g(x)-g'(x)z|\ \mu(x,\dd z)
	&\leq \const \ee^{px}\int_{|z|>1} (\ee^{pz}+1 + |z|)\ \mu(x,\dd z)
	\leq  \const \ee^{p x} , \label{eq:b3}
\end{align}
where each $\const$ does not depend on $x$.  Combining \eqref{eq:b1}, \eqref{eq:b2}, \eqref{eq:b3}, and the bounds on $g'$ and $g''$, we have
\begin{align}
\sup_{t\in[0,T]}\Big|\int_0^t\Ac g(X_{s-})\dd\tau_s\Big|\leq
\int_0^T |\Ac g(X_{s-})|\dd\tau_s \leq C\tau_T\sup_{t\in[0,T]} \ee^{pX_t} ,
\end{align}
which is integrable in case \eqref{eq:gtlin} because $\Eb\tau_T<\infty$, and in case
\eqref{eq:gtexp} by Lemma \ref{lem:supexp}.
The remaining component of $\Gamma$ has magnitude
\begin{align}
|g(X_t)-g(X_0)|
	&\leq \left\{
		\begin{aligned}
			&C(1+|X_t|) & \text{in case \eqref{eq:gtlin}} , \\
			&C (1+\ee^{pX_t}) &\text{in case \eqref{eq:gtexp}} ,
		\end{aligned} \right.  \label{eq:Gdiff}
\end{align}
which has integrable supremum by Lemmas \ref{lem:supx} and \ref{lem:supexp}.
\end{proof}

In conclusion, we relate $\Eb \, [\log F]_T$ to the value of a European-style contract:

\begin{theorem}
\label{thm:main}
Assume that the forward price $F$, 
the log-price 
$X$, and 
the clock
$\tau$ satisfy the assumptions of Section \ref{sec:assumptions}.  Assume that $G$ is a sum of finitely many $C^{1+}(\Rb)$ functions, each of which satisfies \eqref{eq:gtlin} or \eqref{eq:gtexp}, and that $\Ac \, G$ satisfies (for a.e.\ $x$)
\begin{align}
\Ac \, G(x)
	&=	a^2(x) + \int_\Rb z^2 \mu(x,\dd z). \label{eq:AG=qv}
\end{align}
Then $G$ \textit{prices the variance swap}, meaning that
\begin{align}
\Eb \[ \log F \]_T
	&=	 \Eb \, G( \log F_T ) - G(\log F_0) . \label{eq:payoff.0}
\end{align}
Thus, if $\Pb$ is a martingale measure for VS and $G$ contracts, then the fair strike of the VS (equivalently: the forward price of the floating leg of the VS) is \eqref{eq:payoff.0}.
\end{theorem}

\begin{remark}
The sum of finitely many functions is more general than a single function; for instance, $G$ may be the sum of two functions, one satisfying \eqref{eq:gtexp} for some $p>0$, and the other for some $p<0$.
\end{remark}

\begin{remark}\label{rem:addconstants}
Functions $G$ that satisfy the conditions of Theorem \ref{thm:main}, and therefore price the VS, are not unique.  Indeed, if $G$ does, then so does $G(\cdot)+C_0+C_1\exp(\cdot)$, where $C_0, C_1$ are any constants.   Adding the latter two terms does not
affect the valuation $\Eb G(\log F_T)-G(\log F_0)$, because $\Eb F_T=F_0$.  
\end{remark}

\begin{proof}[Proof of Theorem \ref{thm:main}]
We have
\begin{align}
\Eb \, [X]_T
	&=	\Eb \Big(\int_0^{T}  a^2(X_{t}) \dd\tau_t + \int_0^{T} \int_\Rb z^2 N(\dd t,\dd z)\Big)\label{eq:qvjumps}\\
	&=	\Eb \int_0^{T} \( a^2(X_{t-}) + \int_\Rb z^2 \mu(X_{t-},\dd z) \) \dd\tau_t \label{eq:qvcompensator} \\
	&=	\Eb \int_0^{T} \Ac G(X_{t-})\dd\tau_t\label{eq:useoide}\\
	&=	\Eb \, G(X_{T}) - G(X_0) \vphantom{\int_0^{\tau_T^n}}\label{eq:dynkin}
\end{align}
by \citet[Theorems I.4.52 and II.1.8]{jacod1987limit},
equation \eqref{eq:AG=qv} and Lemma \ref{lem:martingale}.
\end{proof}

Theorem \ref{thm:main} allows us to value a VS relative to the $T$-maturity implied volatility smile as follows:
\begin{align}
\underbrace{ \Eb \[ \log F \]_T }_{\text{A}}	
	&=	\underbrace{ \Eb \, G( \log F_T )  }_{\text{B}} - \underbrace{ G(\log F_0) }_{\text{C}} . \label{eq:payoff}
\end{align}
A $=$ the amount agreed upon at time $0$ to pay at time $T$ when taking the long side of a variance swap.\\
B $=$ the value of a European contract with payoff $G(\log F_T)$.\\
C $=$ the value of $G(\log F_0)$ zero-coupon bonds.\\
As shown in \cite{carrmadan1998}, if $h$ is a difference of convex functions, then for any $\kappa \in \Rb^+$ we have
\begin{align}
h(F_T)
	&=		h(\kappa) + h'(\kappa) \Big( (F_T - \kappa)^+ - (\kappa-F_T)^+ \Big)
				+ \int_0^\kappa h''(K)(K-F_T)^+ \dd K + \int_\kappa^\infty h''(K)(F_T-K)^+ \dd K . \label{eq:h.temp}
\end{align}
Here, $h'$ is the left-derivative of $h$, and $h''$ is the second derivative, which exists as a generalized function.
Taking expectations,
\begin{align}
\Eb \, h(F_T)
	&=		h(\kappa) + h'(\kappa) \Big( C(T,\kappa) - P(T,\kappa) \Big)
				+ \int_0^\kappa h''(K)P(T,K)\dd K + \int_\kappa^\infty h''(K)C(T,K) \dd K , \label{eq:h.call.put}
\end{align}
where $P(T,K)$ and $C(T,K)$ are, respectively, the prices of put and call options on $F$ with strike $K$ and expiry $T$.  Knowledge of $F_0$ and the $T$-expiry smile implies knowledge of the initial prices of $T$-expiry European options at all strikes $K > 0$.  Thus the quantity B in \eqref{eq:payoff} is uniquely determined from the $T$-expiry volatility smile by applying \eqref{eq:h.call.put} to $h = G \circ \log$, assuming one can determine the function $G$.  Therefore, to price a VS relative to co-terminal calls and puts, what remains is to find a solution $G$ of the OIDE  \eqref{eq:AG=qv}.

%
%

\section{Examples}
\label{sec:examples}

In this section we provide examples, in the setting of Section \ref{sec:tcsde}, of local variance and L\'evy kernel pairs $(a^2,\mu)$, such that solutions $G$ of OIDE \eqref{eq:AG=qv} can be obtained explicitly.  
In one of the examples, moreover, we investigate the ratio between the values of the VS and the log contract.



\subsection{Constant relative jump intensity}
\label{sec:proportional}
\begin{theorem}
\label{case:proportional}
Assume the local variance $a^2(x)$ and L\'evy kernel $\mu(x,\dd z)$ are of the form
\begin{align}
a^2(x)
	&= 		\gam^2(x) \, \sig^2 , &
\mu(x,\dd z)
	&=		\gam^2(x) \, \nu(\dd z) ,
\end{align}
where $\sig \geq 0$ is a constant, $\nu$ is a L\'evy measure, and $\gam$ is a positive bounded Borel function.
Assume $\Eb\tau_T<\infty$.  Then
\begin{align}
G(x)
	&:=		-Q \, x  , \label{eq:G.prop}
\end{align}
prices the variance swap, where
\begin{align}
Q	&:=		\frac{\sig^2 + \moment_2}{\sig^2/2 + \expmoment_0} , &
\expmoment_0	&:=		\int_\Rb ( \ee^z-1-z ) \nu(\dd z) , &
\moment_2	&:=		\int_\Rb  z^2 \nu(\dd z), \label{eq:Q}
\end{align}
\end{theorem}

\begin{proof}
One can verify directly that $G$ in \eqref{eq:G.prop} satisfies \eqref{eq:gtlin} and \eqref{eq:AG=qv}.
\end{proof}
\begin{remark}
In particular, the constant $Q$
in two extreme cases 
is
as follows
\begin{align}
\text{No Jumps $(\nu\equiv 0)$}:&& Q&=2 , \label{eq:no.jumps}\\
\text{Pure Jumps $(\sig=0)$}:   && Q&=\moment_2/\expmoment_0 . \label{eq:pure.jumps}
\end{align}
\end{remark}

\begin{remark}
\label{rmk:tc}
Dynamics of this form arise by time-changing a \levy process $Y_u$ using the clock
\begin{equation}
\tau_t := \inf\Big\{u\geq 0: \int_0^u \frac{1}{\gamma^2(Y_v)}\dd v\geq t\Big\} .
\end{equation}
See, for instance,
\citet[Proposition 11.6.1]{kuchler1997exponential}.
Thus the payoff function \eqref{eq:G.prop} in this case should, and indeed does, match the payoff function obtained by \cite{carr2011variance} for time-changed L\'evy processes.
\end{remark}


\subsection{Fractional linear relative jump intensity}
\label{sec:fractional}
Let $\alpha, \beta, z_0\in\Rb$ satisfy
\begin{align}
z_0
	&<	0	, &
	&\text{and}&
0
	&<\beta<1-\frac{2(\ee^{z_0}-z_0-1)}{z_0^2}.
\end{align}
Let
\begin{align}
\gamma_3
	&:= -\frac{\alpha}{2\beta}-\frac{1}{\beta} , &
\gamma_0
	&:= -\frac{\alpha}{2\beta}+\frac{z_0^2}{2(\ee^{z_0}-z_0-1)}(1-\frac{1}{\beta}) < \gamma_3 .
\end{align}
Let $\gamma_1$ and $\gamma_2$ satisfy $\gamma_0<\gamma_1<\gamma_2<\gamma_3$.

Define the $C^1$ function
\begin{align}
G(x)
	&:=	\left\{ \begin{aligned}
				&\alpha \gamma_1+\beta \gamma_1^2+(x-\gamma_1)(\alpha+2\beta \gamma_1) 	
				& x<\gamma_1 , \\
				&\alpha x+\beta x^2
				& \gamma_1\leq x\leq \gamma_2 , \\
				&\alpha \gamma_2+\beta \gamma_2^2+(x-\gamma_2)(\alpha+2\beta \gamma_2)
				& x>\gamma_2 .
			\end{aligned} \right. \label{eq:gdef}
\end{align}
We can and do take $\partial^2 G(x) = 2\beta\mathbf{1}_{x\in[\gamma_1,\gamma_2]}$ in the sense of Theorem \ref{thm:main}.

Let $a$ be a positive, bounded, Borel function, and let
\begin{align}
c(x)
	&:=	\frac{a^2(x)}{2} \times \frac{\partial^2 G(x) - \partial G(x) - 2 }{G(x)-G(x+z_0)+(\ee^{z_0}-1)\partial G(x)+z_0^2} . \label{eq:cdef}
\end{align}

\begin{lemma}
The function $c$ is positive and bounded.
\end{lemma}
\begin{proof}  To show that the denominator $G(x)-G(x+z_0)+(\ee^{z_0}-1)\partial G(x)+z_0^2$ from \eqref{eq:cdef} has a positive lower bound, first note that
\begin{align}\label{eq:denomineq}
(\ee^{z_0}-1-z_0) G'(\gamma_2)+z_0^2 > (\ee^{z0}-1-z_0) G'(\gamma_1)+z_0^2 = (\ee^{z0}-1-z_0)(\alpha+2\beta \gamma_1)+z_0^2 > \beta z_0^2 ,
\end{align}
where the first two expressions are the denominator for $x>\gamma_2-z_0$ and $x<\gamma_1$ respectively.

For $x\in(\gamma_2,\gamma_2-z_0)$, the denominator is bounded below by
$-\half\sup_{x\in\Rb}|\partial^2G(x)|z_0^2+(\ee^{z0}-1-z_0) G'(\gamma_2)+z_0^2$, so just subtract $\beta z_0^2$ from \eqref{eq:denomineq}.
For $x\in(\gamma_1,\gamma_2)$ the denominator is bounded below by
\begin{align}
(1-\beta)z_0^2+(\alpha+2\beta x)(\ee^{z_0}-1-z_0) > (1-\beta)z_0^2+(\alpha+2\beta\gamma_1)(\ee^{z_0}-1-z_0) > 0.
\end{align}

Next, to show that the numerator $\partial^2G-\partial G-2$ from \eqref{eq:cdef} is positive and bounded, we verify in three intervals.
For $x\in(\gamma_1,\gamma_2)$, the numerator is
$2\beta -\alpha - 2-2\beta x > 2\beta -\alpha - 2-2\beta \gamma_3 = 2\beta > 0$, and is moreover bounded above.
In the other two intervals, the result follows from
\begin{align}
-\alpha-2\beta \gamma_1 - 2 > -\alpha - 2\beta \gamma_2-2 > -\alpha - 2\beta \gamma_3-2 = 0,
\end{align}
where the first two expressions are the numerator for $x\leq \gamma_1$ and $x\geq \gamma_2$ respectively.
\end{proof}

\begin{theorem}
Assume the local L\'evy kernel $\mu$ is a point mass at $z_0$ with weight $c(x)$:
\begin{align}
\mu(x,\cdot)
	&=		c(x)\delta_{z_0},
\end{align}
where $c$, $G$, and the local variance $a^2$ are related by \eqref{eq:gdef} and \eqref{eq:cdef}.	
Assume $\Eb\tau_T<\infty$.
Then $G$ prices the variance swap.
\end{theorem}
\begin{proof}
We have that $G$ satisfies \eqref{eq:gtlin} and, by \eqref{eq:cdef}, the OIDE \eqref{eq:AG=qv}.
\end{proof}

\begin{remark}
We describe these dynamics as ``fractional linear relative jump intensity'' because, for $x\in(\gamma_1-z_0,\gamma_2)$, the relative jump intensity
\begin{align}
\frac{c(x)}{a^2(x)}=\frac{\beta-\alpha/2-1-\beta x}{\alpha(\ee^{z_0}-1-z_0)+(1-\beta)z_0^2+2\beta x(\ee^{z_0}-1-z_0)}
\end{align}
is a ratio of polynomials linear in the underlying log-price.
\end{remark}

\subsection{\levy mixture with state-dependent weights}
\label{sec:unified}
Assume the local variance $a^2(x)$ and L\'evy kernel $\mu(x,\dd z)$ are of the form
\begin{align}
a^2(x)
	&=		\alpha \sig_0^2(x) + \del \beta \sig_1^2(x) , &
\mu(x,\dd z)
	&=		\frac{\sig_0^2(x)}{2} \nu_0(\dd z)
				+ \del \frac{\sig_1^2(x)}{2} \nu_1(\dd z) , &
\frac{\sig_1^2(x)}{\sig_0^2(x)}
	&=	\ee^{cx}	=:	\ee_c(x) ,	\label{eq:setting.d}
\end{align}
where $\alpha,\beta,\del \geq 0$ and $\nu_0, \nu_1$ are \levy measures with
\begin{align}
\int_\Rb \Big|\ee^{\lam z} - 1 + (1-\ee^z) \lam \Big|\ \nu_i(\dd z)
	&< \infty , &
\forall \lam
	&\in \Cb , &
i
	&\in \{0, 1\} .  \label{eq:integrability}
\end{align}
Let us first derive a candidate solution $G$ to \eqref{eq:AG=qv} from an ansatz and then verify the validity of the solution.

Inserting the expressions for $a^2$ and $\mu$ 
from
\eqref{eq:setting.d} into \eqref{eq:AG=qv} and dividing by $\frac{1}{2}\sig_0^2(x)$, we have
\begin{align}
( \Ac_0  + \del \ee_c \Ac_1 ) G
	&=	I_0 + \del \ee_c I_1 , \label{eq:G.OIDE}
\end{align}
where $I_0$ and $I_1$ are constants defined 
by
\begin{align}
I_0
	&=	2 \alpha + \int_\Rb z^2\ \nu_0(\dd z), &
I_1
	&=	2 \beta + \int_\Rb z^2\ \nu_1(\dd z) ,
\end{align}
and, using the notation of \eqref{eq:A.first}, the operators $\Ac_0$ and $\Ac_1$ are given by
\begin{align}
\Ac_0
	&=	\alpha \( \d^2 - \d \) + \int_\Rb \Big(\ee^{z\d} - 1 +(1-\ee^z) \d \Big) \nu_0(\dd z)  , \\
\Ac_1
	&=	\beta \( \d^2 - \d \) + \int_\Rb \Big(\ee^{z\d} - 1 +(1-\ee^z) \d \Big) \nu_1(\dd z) .
\end{align}
Assume a solution $G$ of \eqref{eq:G.OIDE} has a power series expansion in $\del$:
\begin{align}
G
	&=	\sum_{n=0}^\infty \del^n G_n , \label{eq:G.expand}
\end{align}
where the functions $\{ G_n \}_{n \geq 0}$ are, at this point, unknown.  Inserting the expansion \eqref{eq:G.expand} into OIDE \eqref{eq:G.OIDE} and collecting terms of like order in $\del$ results in the following sequence of nested OIDEs:
\begin{align}
\Oc(1):&&
\Ac_0 G_0
	&=	I_0 , \label{eq:G0.OIDE}  \\
\Oc(\del):&&
\Ac_0 G_1 + \ee_c \Ac_1 G_0
	&=	\ee_c I_1 , \label{eq:G1.OIDE} \\
\Oc(\del^n):&&
\Ac_0 G_n + \ee_c \Ac_1 G_{n-1}
	&=	0 , &
n
	&\geq 2 . \label{eq:Gn.OIDE}
\end{align}
Noting that
\begin{align}
\Ac_0 \ee_{\lam}
	&=	\phi_\lam \ee_{\lam} , &
\phi_\lam
	&=	\alpha \( \lam^2 - \lam \) + \int_\Rb \Big(\ee^{\lam z} - 1 + (1-\ee^z) \lam \Big) \nu_0(\dd z) , &
\forall \lam
	&\in \Cb , \\
\Ac_1 \ee_{\lam}
	&=	\chi_\lam \ee_{\lam} , &
\chi_\lam
	&=	\beta \( \lam^2 - \lam \) + \int_\Rb \Big(\ee^{\lam z} - 1 + (1-\ee^z) \lam \Big) \nu_1(\dd z), &
\forall \lam
	&\in \Cb ,
\end{align}
one can easily verify, by direct substitution 
into \eqref{eq:G1.OIDE},
a solution $G_0$ given by
\begin{align}
G_0(x)
	&:=	- Q_0 x , &
Q_0
	&:=	\frac{2 \alpha + \int_\Rb  z^2 \nu_0(\dd z)}{ \alpha + \int_\Rb (\ee^z - 1 - z )\nu_0(\dd z)} , \label{eq:G0} 
\end{align}
and solutions $\{ G_n \}_{n \geq 1}$ given, for $c\neq 0$, by
\begin{align}
G_n(x)
	&:=	Q_1 \frac{\ee_{nc}(x)}{\phi_{nc}} \prod_{k=1}^{n-1} \frac{-\chi_{kc}}{\phi_{kc}}, &
Q_1
	&:=	2 \beta + \int_\Rb z^2 \nu_1(\dd z) - Q_0 \Big( \beta + \int_\Rb (\ee^z - 1 -z ) \nu_1(\dd z) \Big). \label{eq:Gn}
\end{align}
Thus we have a formal series expansion, defined by \eqref{eq:G.expand}, \eqref{eq:G0} and \eqref{eq:Gn}, for a function $G$ that solves OIDE \eqref{eq:AG=qv}.  The following conditions
suffice for validity of this expansion.
\begin{theorem}
\label{thm:unified}
Assume that the local variance $a^2(x)$ and L\'evy kernel $\mu(x,\dd z)$ are given by \eqref{eq:setting.d}.  Assume further that $\nu_0$ and $\nu_1$ satisfy \eqref{eq:integrability} and $c\neq 0$ and
\begin{align}
	\lim_{n \to \infty } \frac{ \beta n^2c^2 + \int_\Rb \nu_1(\dd z)\(\ee^{ncz}-1 + (1-\ee^z) nc\)}
														{ \alpha n^2c^2 + \int_\Rb \nu_0(\dd z)\(\ee^{(n+1)cz} - 1 + (1-\ee^z) (n+1)c\) }
	&= 0 . \label{eq:nu.condition}
\end{align}
Then the function $G$
is well-
defined on $\mathbb{R}$ by \eqref{eq:G.expand} with \eqref{eq:G0} and \eqref{eq:Gn}, 
and  
solves OIDE \eqref{eq:AG=qv}.
\end{theorem}
\begin{proof}
The summation in \eqref{eq:G.expand} 
can be written as 
\begin{align}
		&		- Q_0 x  + Q_1 \sum_{n=1}^\infty a_n u^n(x), &
\text{where }a_n
		=		\frac{1}{\phi_{nc}} \prod_{k=1}^{n-1} \frac{-\chi_{kc}}{\phi_{kc}}, \text{ and }
u(x)
		=		\del \ee_c(x). \label{eq:G.power.series}
\end{align}
The infinite sum is a power series in $u$, 
with coefficients $\{a_n\}_{n \geq 1}$ satisfying, by \eqref{eq:nu.condition},
\begin{align}
\lim_{n \to \infty} \frac{a_{n+1}}{a_n}
	&=		\lim_{n \to \infty} \frac{-\chi_{nc}}{\phi_{(n+1)c}} = 0 , \label{eq:asymptotic.1}
\end{align}
which implies that the sum in \eqref{eq:G.power.series} has infinite radius of convergence,
and $G$ is well-defined on $\mathbb{R}$ by \eqref{eq:G.expand}, with \eqref{eq:G0} and \eqref{eq:Gn}. 
As every power series can be differentiated and integrated term-by-term within its radius of convergence, 
$G$ solves OIDE \eqref{eq:AG=qv}.
\end{proof}
\begin{remark}
If $\alpha=0$, $\beta>0$, $\nu_1 \equiv 0$, and $c>0$ (respectively, $c<0$), then any L\'evy measure $\nu_0$ with support on the positive (resp.\ negative) axis will satisfy \eqref{eq:nu.condition}.
\end{remark}
\begin{remark}
\label{rmk:negative.jumps}
If $\alpha>0$, $\beta=0$, 
$\nu_0 \equiv 0$, and $c>0$ (respectively, $c<0$), then a L\'evy measure $\nu_1$ will satisfy \eqref{eq:nu.condition} only if the support of $\nu_1$ lies strictly within the negative (resp.\ positive) axis.
\end{remark}
\begin{remark}
\label{rmk:perturbation}
In the particular case where the forward price $F$ is a time-change of an exponential L\'evy process with variance {$\alpha$} and L\'evy measure $\nu_0$, the function $G_0$ prices the VS.  In
the more general class of models in \eqref{eq:setting.d}, which can be seen as a regular $\delta$-perturbation around the time-changed exponential L\'evy case, the candidate function $G$ for pricing the VS by Theorem \ref{thm:main}  
becomes, by \eqref{eq:G.power.series}, a $\delta$-perturbation around $G_0$.
\end{remark}


\noindent
In Figures \ref{fig:6} and \ref{fig:7}, using a variety of different model parameters, we plot
\begin{align}
h(F_T)
	&:=	G(\log F_T) - G(\log F_0) + A (F_T - F_0) , &
A
	&=	{ \frac{-1}{F_0} G'(\log F_0) } . \label{eq:h}
\end{align}
as a function 
of $F_T$, where $G$ is defined by \eqref{eq:G.expand}, \eqref{eq:G0} and \eqref{eq:Gn}.  
Note that, if $G$ prices the VS, then $h$ prices the VS for any constant $A$.  The particular value of $A$ in \eqref{eq:h} ensures that $h'(F_0)=0$.


\subsubsection{Ratio of the VS value to the $\log$ contract value}
Although the purpose of this paper is to compute the value of a VS relative to the $G$ contract (and to solve for $G$),
it is interesting to compute the ratio of the value of the VS to the value of a European $\log$ contract.  To this end, for a function $G$ that prices a VS, 
let
\begin{align}
\Qc(T,F_0)
	&:= \frac{\Eb \, G(\log F_T) - G(\log F_0)}{-\Eb \log(F_T/F_0)} 
	=		\frac{\Eb \, [\log F]_T}{-\Eb \log(F_T/F_0)} . \label{eq:Qc}
\end{align}
In \cite{carr2011variance} the authors find that if $F_t=\exp(\widehat{Y}_{\tau_t})$ where $\widehat{Y}$ is a \textit{L\'evy process}, then the ratio $\Qc(T,F_0)$ is a constant $Q$ which is \textit{independent} of the initial value $F_0$ of the underlying and the time to maturity $T$ (see Theorem \ref{case:proportional} and Remark \ref{rmk:tc} of Section \ref{sec:proportional}).  This is in contrast to empirical results from the same paper, which show in a study of S\&P500 data that the ratio $\Qc(T,F_0)$ is not constant.  In the more general time-changed Markov setting considered in the present paper, the ratio $\Qc(T,F_0)$ can (in general) depend on the 
current value $F_0$ of the underlying and the time to maturity $T$.  Below, we derive a formal approximation for the ratio $\Qc(T,F_0)$ for one specific example of $(a^2,\mu)$ which is of the form \eqref{eq:setting.d}.
\begin{assumption}
\label{ass:1}
Throughout this section, we assume $F_t = \exp(Y_{\tau_t})$ where $\tau$ is a continuous time change \textit{independent} of $Y$ and the Laplace transform $L(t,\lam) :=	\Eb \, \ee^{\tau_t \lam}$ is known.  Let the Markov process $Y$ have local variance $a^2(x)$ and L\'evy kernel $\mu(x,\dd z)$ of the form
\eqref{eq:setting.d} with
\begin{align}
\alpha
	&=1 , &
\beta
	&=	0 , &
\sig_0^2(x)
	&=	2 \om^2 , &
\sig_1^2(x)
	&=	2 \om^2 \ee_c(x) , &
\nu_0
	&\equiv 0 , &
\nu_1
	&\equiv \nu ,
\end{align}
where $\om, c > 0$.  Assume moreover that the L\'evy measure $\nu$ satisfies the conditions of Theorem \ref{thm:unified}.   Thus, the function $G$, defined by \eqref{eq:G.expand}, \eqref{eq:G0} and \eqref{eq:Gn}, 
solves
 \eqref{eq:AG=qv} .  In accordance with Remark \ref{rmk:negative.jumps}, jumps must be downward, i.e., $\nu(\Rb^+)=0$.
\end{assumption}
\noindent
We compute an approximation for $\Qc(T,F_0)$ in three \textbf{Steps}, 
described below.
\begin{step}
\textit{Derive an approximation for $u(t,x;\varphi):=\Eb_x \, \varphi(Y_{t})$.}\\
Formally, the function $u$ satisfies the Kolmogorov backward equation
\begin{align}
( -\d_t + \Ac ) u
	&= 		0 , &
u(0,\cdot;\varphi)
	&=		\varphi, \label{eq:u.pde}
\end{align}
where $\Ac$, the generator of $Y$, is given by
\begin{align}
\Ac
	&=		\om^2 \Ac_0 + \del \ee_c \om^2 \Ac_1 . \label{eq:A.expand}
\end{align}
Now, suppose that the function $u$ has a power series expansion in $\del$
\begin{align}
u
	&=	\sum_{n=0}^\infty \del^n u_n , \label{eq:u.expand}
\end{align}
where the functions $\{u_n\}_{n \geq 0}$ are unknown.
Inserting expressions \eqref{eq:A.expand} and \eqref{eq:u.expand} into \eqref{eq:u.pde} and collecting terms of like powers of $\del$, we obtain a sequence of nested partial integro-differential equations (PIDEs) for the unknown functions $\{u_n\}_{n \geq 0}$
\begin{align}
\Oc(1):&&
( -\d_t + \om^2 \Ac_0 ) u_0
	&= 		0 , &
u_0(0,\cdot;\varphi)
	&=		\varphi, \label{eq:u0.pde} \\
\Oc(\del^n):&&
( -\d_t + \om^2 \Ac_0 ) u_n
	&= 		- \ee_c \om^2 \Ac_1 u_{n-1} , &
u_n(0,\cdot;\varphi)
	&=		0 , &
n
	&\geq 1 . \label{eq:un.pde}
\end{align}
The solution to this nested sequence of PIDEs is given in \citet[Equation ({5.2})]{lorig-jacquier-1}. We have
\begin{align}
u_n(t,x;\varphi)
	&=	\int_\Rb \( \sum_{k=0}^n \frac{\ee^{t \om^2 \phi_{\ii \lam + kc}}\ee_{\ii \lam + nc}(x)}
      {\prod_{j \neq k}^n ( \om^2 \phi_{\ii \lam + kc} - \om^2 \phi_{ \ii\lam + j c})}\)
      \( \prod_{k=0}^{n-1} \om^2 \chi_{\ii \lam + k c}\) \varphih(\lam)  \dd \lam, \label{eq:un}
\end{align}
where an empty product is defined to equal one $\prod_{k=0}^{-1}(\cdots):=1$ and $\varphih$ denotes the 
distributional generalization of the Fourier transform defined for integrable functions $\varphi$ by
\begin{align}
\varphih(\lam)
	&
:= \frac{1}{2\pi}\int_\Rb \varphi(x) \ee^{- \ii \lam x} \dd x.
\end{align}
Inserting expression \eqref{eq:un} into the sum \eqref{eq:u.expand} and truncating at order $N$ yields
$\ub_N$, our \textit{$N^\textup{th}$ order approximation of $u$}.  Explicitly, 
\begin{align}
\ub_N(t,x;\varphi)
	&:= \sum_{n=0}^N \del^n u_n(t,x;\varphi) \\
	&=	\int_\Rb \sum_{n=0}^N \del^n \( \sum_{k=0}^n \frac{\ee^{t \om^2 \phi_{\ii \lam + kc}}\ee_{\ii \lam + nc}(x) }
      {\prod_{j \neq k}^n ( \om^2 \phi_{\ii \lam + kc} - \om^2 \phi_{ \ii\lam + j c})}\)
      \( \prod_{k=0}^{n-1} \om^2 \chi_{\ii \lam + k c}\) \varphih(\lam) \dd \lam. \label{eq:uN}
\end{align}
\end{step}
\begin{step}
\textit{Derive an approximation for $v(t,x;\varphi):=\Eb_x \, \varphi(Y_{\tau_t})$.}\\
Using the independence of $\tau$ and $Y$, we have
\begin{align}
v(t,x;\varphi)
	&:= \Eb_x \varphi(Y_{\tau_t})
	=		\Eb_x \, \Eb_x[ \varphi(Y_{\tau_t}) | \tau_t ]
	=		\Eb \, u(\tau_{t},x;\varphi) . \label{eq:conditioning}
\end{align}
Replacing the function $u$ in \eqref{eq:conditioning} with $\ub_N$ yields $\vb_N$, our \textit{$N^\textup{th}$ order approximation of $v$}.  Explicitly, 
\begin{align}
\vb_N(t,x;\varphi)
	&:=	\Eb \, \ub_N(\tau_{t},x;\varphi)\\
	&=	\int_\Rb \sum_{n=0}^N \del^n \( \sum_{k=0}^n \frac{L(t,\om^2 \phi_{\ii \lam + kc})\ee_{\ii \lam + nc}(x) }
      {\prod_{j \neq k}^n ( \om^2 \phi_{\ii \lam + kc} - \om^2 \phi_{ \ii\lam + j c})}\)
      \( \prod_{k=0}^{n-1} \om^2 \chi_{\ii \lam + k c}\) \varphih(\lam) \dd \lam , \label{eq:vN}
\end{align}
using equation \eqref{eq:uN} and $\Eb \ee^{\lam \tau_t}=L(t,\lam)$.
\end{step}
\begin{step}
\textit{Derive an approximation for $\Qc(T,F_0)$.}\\
With $G$ as given in Theorem \ref{thm:unified}, we have
\begin{align}
\Qc(T,F_0) 
	&=	\frac{ \Eb G( \log F_T) - G(\log F_0) }{-\Eb \log (F_T/F_0) } \\
	&=	Q_0 + \frac{ \sum_{n=1}^\infty b_n \Big( \Eb \ee_{nc}( \log F_T ) - \ee_{nc}( \log F_0 ) \Big) }{ - \Eb \log F_T + \log F_0 } \\
	&=	Q_0 + \frac{ \sum_{n=1}^\infty b_n \Big( \Eb \ee_{nc}( Y_{\tau_T} ) - \ee_{nc}( \log F_0 ) \Big) }{ - \Eb Y_{\tau_T} + \log F_0 } \\
	&=	Q_0 + \frac{ \sum_{n=1}^\infty b_n \Big( v(T,\log F_0;\ee_{nc}) - \ee_{nc}( \log F_0 ) \Big) }{ - v(T,\log F_0;\Id) + \log F_0 } , \label{eq:ratio} \\
b_n
	&:=	Q_1 \del^n a_n
	=		Q_1 \frac{\del^n}{\phi_{nc}} \prod_{k=1}^{n-1} \frac{-\chi_{kc}}{\phi_{kc}},
\end{align}
where $\Id$ is the identity function $\Id(x) = x$.
Replacing the function $v$ wherever it appears in \eqref{eq:ratio} by $\vb_N$ and truncating the infinite sum at $N$ terms 
produces
$\bar{\Qc}_N(T,F_0)$, our \textit{$N^\textup{th}$ order approximation of $\Qc(T,F_0)$}.  Explicitly, 
\begin{align}
\bar{\Qc}_N(T,F_0)
	&:=	Q_0 + \frac{\sum_{n=1}^N b_n \Big( \vb_N(T,\log F_0;\ee_{nc}) - \ee_{nc}( \log F_0 ) \Big) }{ - \vb_N(T,\log F_0;\Id) + \log F_0 } . \label{eq:Q-bar.N}
\end{align}
The Fourier transforms of the complex exponential $\ee_\gam$ ($\gam \in \Cb$) and the identity function $\Id$, 
as %
needed to compute $\vb_N(T,\log F_0;\ee_{nc})$ and $\vb_N(T,\log F_0;\Id)$ in \eqref{eq:Q-bar.N}, are given by
\begin{align}
\widehat{\ee}_\gam(\lam)
	&=	\del(\lam + \ii \gam) , &
\gam
	&\in \Cb , &
\widehat{\Id}(\lam)
	&=	\ii \del'(\lam) , \label{eq:FTs}
\end{align}
where $\del$ and $\del'$ denote the Dirac delta function and its derivative, 
understood in the sense of distributions.
Inserting \eqref{eq:FTs} into \eqref{eq:vN} and integrating produces closed-form expressions for both $\vb_N(T,\log F_0;\ee_{nc})$ and $\vb_N(T,\log F_0;\Id)$.
\end{step}
\noindent
Figure \ref{fig:ratio.small.jump}  plots
 $\bar{\Qc}_N(T,F_0)$ as a function of $F_0$.

%
%

\section{Conclusion}
\label{sec:conclusion}
In \cite{carr2011variance}, the authors model the forward price as the exponential of a L\'evy process time-changed by a continuous increasing stochastic clock.  In this setting, they show that a variance swap has the same value as a fixed number of European $\log$ contracts.  The exact number of $\log$ contracts that price the variance swap depends only on the dynamics of the driving L\'evy process, irrespective of the time-change.
\par
This paper generalizes the underlying forward price dynamics to time-changed exponential Markov processes, where the background process may have a state-dependent (i.e., local) volatility and L\'evy kernel, and where the stochastic time-change may have arbitrary dependence or correlation with the Markov process.  In the time-changed Markov setting, we prove that the variance swap is priced by a European-style contract whose payoff depends only on the dynamics of the Markov process, not on the time-change.  We explicitly compute the payoff function that prices the variance swap for various driving Markov processes.  When the Markov process is a L\'evy process we recover the results of \cite{carr2011variance}.
\par
For certain Markov processes, we also compute directly from model parameters an approximation for valuation of European-style contracts, 
showing the variation in 
the ratio of the VS value to the $\log$ contract value 
as a function of the 
current level of the underlying.  This is in contrast to \cite{carr2011variance}, who show in the more restrictive time-changed L\'evy process setting that this ratio is constant.

\subsection*{Thanks}
The authors are grateful to Feng Zhang and Stephan Sturm for their helpful comments.

\bibliographystyle{chicago}
\bibliography{BibTeX-Master-New}

\begin{thebibliography}{}

\bibitem[\protect\citeauthoryear{Carr, Lee, and Wu}{Carr
  et~al.}{2012}]{carr2011variance}
Carr, P., R.~Lee, and L.~Wu (2012, 4).
\newblock Variance swaps on time-changed {L\'e}vy processes.
\newblock {\em Finance and Stochastics\/}~{\em 16\/}(2), 335--355.

\bibitem[\protect\citeauthoryear{Carr and Madan}{Carr and
  Madan}{1998}]{carrmadan1998}
Carr, P. and D.~Madan (1998).
\newblock Towards a theory of volatility trading.
\newblock In {\em Volatility: new estimation techniques for pricing
  derivatives}, pp.\  417--427. Risk Books.

\bibitem[\protect\citeauthoryear{\c{C}inlar and Jacod}{\c{C}inlar and
  Jacod}{1981}]{cinlarjacod}
\c{C}inlar, E. and J.~Jacod (1981).
\newblock Representation of semimartingale {M}arkov processes in terms of
  {W}iener processes and {P}oisson random measures.
\newblock In {\em Seminar on Stochastic Processes, 1981}, Volume~1 of {\em
  Progress in Probability and Statistics}, pp.\  159--242. Birkh\"auser Boston.

\bibitem[\protect\citeauthoryear{Cui, Kirkby, and Nguyen}{Cui
  et~al.}{2017}]{CUI2017381}
Cui, Z., J.~L. Kirkby, and D.~Nguyen (2017).
\newblock A general framework for discretely sampled realized variance
  derivatives in stochastic volatility models with jumps.
\newblock {\em European Journal of Operational Research\/}~{\em 262\/}(1), 381
  -- 400.

\bibitem[\protect\citeauthoryear{Dupire}{Dupire}{1993}]{dupire1993model}
Dupire, B. (1993).
\newblock Model art.
\newblock {\em Risk\/}~{\em 6\/}(9), 118--124.

\bibitem[\protect\citeauthoryear{Filipovi{\'c}, Gourier, and
  Mancini}{Filipovi{\'c} et~al.}{2016}]{FILIPOVIC201644}
Filipovi{\'c}, D., E.~Gourier, and L.~Mancini (2016).
\newblock Quadratic variance swap models.
\newblock {\em Journal of Financial Economics\/}~{\em 119\/}(1), 44 -- 68.

\bibitem[\protect\citeauthoryear{Henry-Labord{\`e}re and
  Touzi}{Henry-Labord{\`e}re and Touzi}{2016}]{Henry-Labordere2016}
Henry-Labord{\`e}re, P. and N.~Touzi (2016, 7).
\newblock An explicit martingale version of the one-dimensional {B}renier
  theorem.
\newblock {\em Finance and Stochastics\/}~{\em 20\/}(3), 635--668.

\bibitem[\protect\citeauthoryear{Hobson and Klimmek}{Hobson and
  Klimmek}{2012}]{hobsonklimmek}
Hobson, D. and M.~Klimmek (2012, 10).
\newblock Model-independent hedging strategies for variance swaps.
\newblock {\em Finance and Stochastics\/}~{\em 16\/}(4), 611--649.

\bibitem[\protect\citeauthoryear{Itkin and Carr}{Itkin and
  Carr}{2010}]{itkin2010pricing}
Itkin, A. and P.~Carr (2010).
\newblock Pricing swaps and options on quadratic variation under stochastic
  time change models -- discrete observations case.
\newblock {\em Review of Derivatives Research\/}~{\em 13\/}(2), 141--176.

\bibitem[\protect\citeauthoryear{Itkin and Carr}{Itkin and
  Carr}{2012}]{itkin2012using}
Itkin, A. and P.~Carr (2012).
\newblock Using pseudo-parabolic and fractional equations for option pricing in
  jump diffusion models.
\newblock {\em Computational Economics\/}~{\em 40\/}(1), 63--104.

\bibitem[\protect\citeauthoryear{Jacod and Shiryaev}{Jacod and
  Shiryaev}{1987}]{jacod1987limit}
Jacod, J. and A.~N. Shiryaev (1987).
\newblock {\em Limit theorems for stochastic processes}, Volume 288.
\newblock Springer-Verlag Berlin.

\bibitem[\protect\citeauthoryear{{Jacquier} and {Lorig}}{{Jacquier} and
  {Lorig}}{2013}]{lorig-jacquier-1}
{Jacquier}, A. and M.~{Lorig} (2013).
\newblock The smile of certain {L}{\'e}vy-type models.
\newblock {\em {SIAM} Journal on Financial Mathematics\/}~{\em 4\/}(1),
  804--830.

\bibitem[\protect\citeauthoryear{Kallsen and Shiryaev}{Kallsen and
  Shiryaev}{2002a}]{kallsen2002esscher}
Kallsen, J. and A.~Shiryaev (2002a).
\newblock The cumulant process and esscher's change of measure.
\newblock {\em Finance and Stochastics\/}~{\em 6\/}(4), 397--428.

\bibitem[\protect\citeauthoryear{Kallsen and Shiryaev}{Kallsen and
  Shiryaev}{2002b}]{kallsen2002time}
Kallsen, J. and A.~Shiryaev (2002b).
\newblock Time change representation of stochastic integrals.
\newblock {\em Theory of Probability \& Its Applications\/}~{\em 46\/}(3),
  522--528.

\bibitem[\protect\citeauthoryear{K{\"u}chler and S{\o}rensen}{K{\"u}chler and
  S{\o}rensen}{1997}]{kuchler1997exponential}
K{\"u}chler, U. and M.~S{\o}rensen (1997).
\newblock {\em Exponential families of stochastic processes}, Volume~3.
\newblock Springer Science \& Business Media.

\bibitem[\protect\citeauthoryear{Lepingle and M{\'e}min}{Lepingle and
  M{\'e}min}{1978}]{lepingle1978integrabilite}
Lepingle, D. and J.~M{\'e}min (1978).
\newblock Sur l'int{\'e}grabilit{\'e} uniforme des martingales exponentielles.
\newblock {\em Probability Theory and Related Fields\/}~{\em 42\/}(3),
  175--203.

\bibitem[\protect\citeauthoryear{Lorig, Lozano-Carbass\'e, and
  Mendoza-Arriaga}{Lorig et~al.}{2016}]{lorigmendoza}
Lorig, M., O.~Lozano-Carbass\'e, and R.~Mendoza-Arriaga (2016).
\newblock Variance swaps on defaultable assets and market implied time-changes.
\newblock {\em SIAM Journal on Financial Mathematics\/}~{\em 7\/}(1), 273--307.

\bibitem[\protect\citeauthoryear{Nabil}{Nabil}{2014}]{Nabil}
Nabil, K. (2014).
\newblock Model-independent lower bound on variance swaps.
\newblock {\em Mathematical Finance\/}~{\em 26\/}(4), 939--961.

\bibitem[\protect\citeauthoryear{Neuberger}{Neuberger}{1990}]{neuberger}
Neuberger, A. (1990).
\newblock Volatility trading.
\newblock {\em Working paper: London Business School\/}.

\bibitem[\protect\citeauthoryear{Protter}{Protter}{2004}]{protter}
Protter, P. (2004).
\newblock {\em Stochastic integration and differential equations}, Volume~21.
\newblock Springer Verlag.

\bibitem[\protect\citeauthoryear{Wendong and Kuen}{Wendong and
  Kuen}{2014}]{zheng-kwok}
Wendong, Z. and K.~Y. Kuen (2014).
\newblock Closed form pricing formulas for discretely sampled generalized
  variance swaps.
\newblock {\em Mathematical Finance\/}~{\em 24\/}(4), 855--881.

\end{thebibliography}

%
%

\clearpage
\begin{figure}
\centering
\includegraphics[width=.95\textwidth,height=.37\textheight]{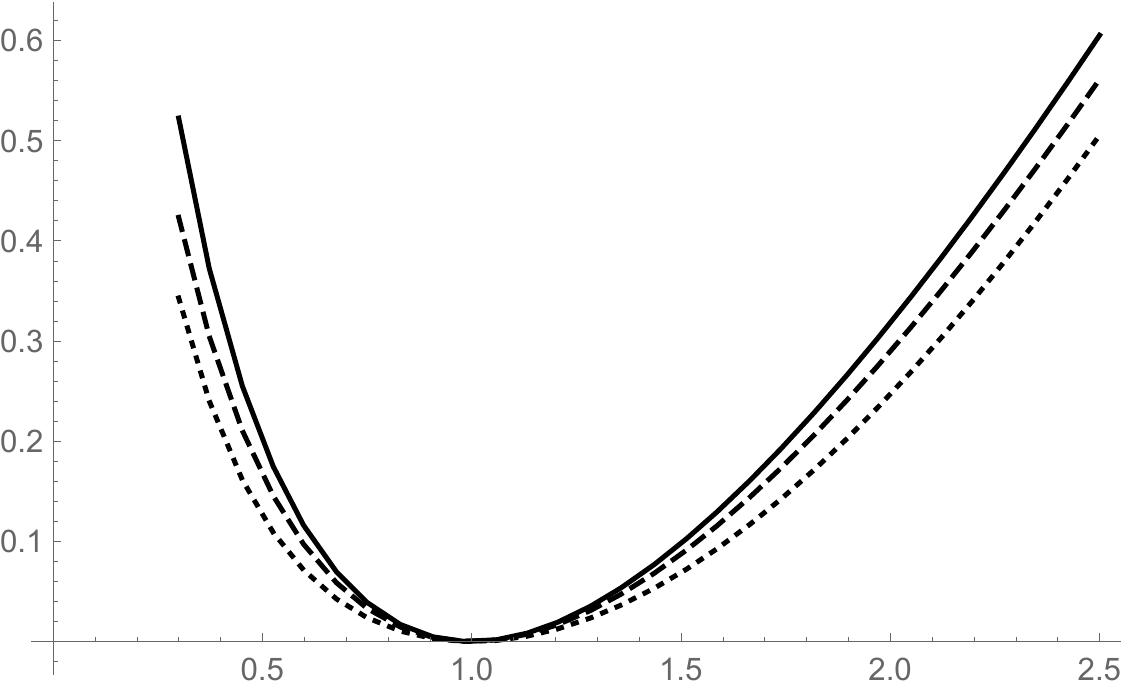}
\caption{
In this figure, we set $F_0=1$, $\alpha = 1$, $\beta = 0$, $\del = 0.35$, $\nu_0 \equiv 0$, $\nu_1 = \del_{z_0}$ where $z_0 = 2.5$,
and we plot $h(F_T)$ as a function of $F_T$ with $c = 0$ (solid), $c=-1$ (dashed) and $c=-2$ (dotted).
Note that, when $c=0$, we are in the setting of Section \ref{sec:proportional} and thus $h$ is a $\log$ contract plus an affine function.
}
\label{fig:6}
\end{figure}

\begin{figure}
\centering
\includegraphics[width=.95\textwidth,height=.37\textheight]{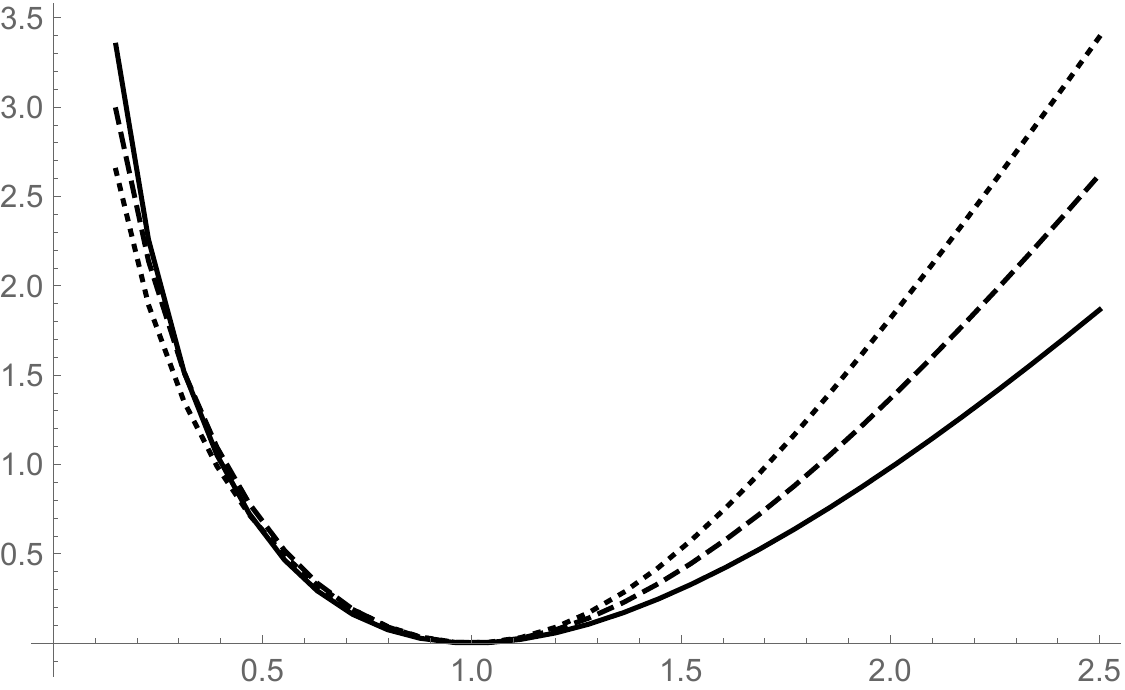}
\caption{
In this figure, we set $F_0=1$, $\alpha = 1$, $\beta = 0$, $\del = 1$, $\nu_0 \equiv 0$, $\nu_1 = \del_{z_0}$ where $z_0 = -2.5$,
and we plot $h(F_T)$ as a function of $F_T$ with $c = 0$ (solid), $c=2$ (dashed) and $c=4$ (dotted).
Note that, when $c=0$, we are in the setting of Section \ref{sec:proportional} and thus $h$ is a $\log$ contract plus an affine function.
}
\label{fig:7}
\end{figure}

\clearpage

\begin{figure}
\centering
\includegraphics[width=.95\textwidth]{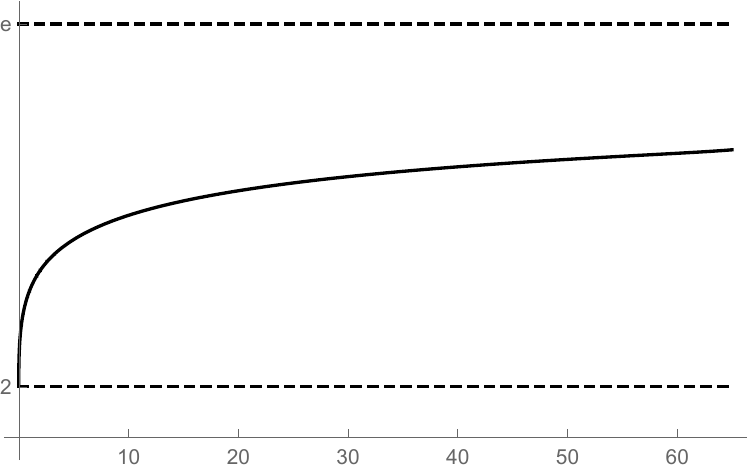}
\caption{A plot of $\bar{\Qc}_N(T,F_0)$, our $N^\textup{th}$ order approximation of $\Qc(T,F_0):=\frac{\Eb \, [\log F]_T}{ - \Eb \, \log(F_T/F_0)}$ as a function of $F_0$ (solid line).  In this plot, the forward price is given by $F_t = \exp(Y_t)$ (i.e., no time-change) and the Markov process $Y$ has local variance $a^2(x) = 2 \om^2$ and L\'evy kernel $\mu(x,\dd z)=\del \om^2 \ee^{cx}\nu(\dd z)$ where $\nu = \del_{z_0}$.  We use the following parameters: $c=0.395$, $\del=1.0$, $\om=0.3$, $z_0=-1.0$ and $T=1.0$.  We fix $N=35$.  Note that as $F_0 \to 0$, the jump intensity goes to zero: $\del\om^2 F_0^c\to 0$.  Accordingly, as $F_0 \to 0$ the ratio $\frac{\Eb \, [\log F]_T}{ - \Eb \, \log(F_T/F_0)} \to 2$, which is what one would expect for a forward price process that experiences no jumps (see equation \eqref{eq:no.jumps}).  As $F_0 \to \infty$ and the jump-intensity increases, we expect the ratio $\frac{\Eb \, [\log F]_T}{ - \Eb \, \log(F_T/F_0)} \to \moment_2/\expmoment_0 = \ee$, which is the corresponding ratio for a pure-jump L\'evy-type process (see equation \eqref{eq:pure.jumps}).  Note that if the Markov process $Y$ were a L\'evy process (i.e., with constant variance coefficient and L\'evy measure), as in \cite{carr2011variance}, the ratio $\frac{\Eb \, [\log F]_T}{ - \Eb \, \log(F_T/F_0)}$ would be a constant \textit{independent} of $F_0$.}
\label{fig:ratio.small.jump}
\end{figure}

\end{document}